\documentclass[12pt,a4paper]{article}

\usepackage{amsfonts,amsmath,amssymb,amsthm}
\usepackage{latexsym,amscd}
\usepackage{amsbsy}
\usepackage{CJK}
\usepackage{indentfirst,mathrsfs}
\usepackage{amsfonts,amsmath,amssymb,amsthm}
\usepackage{mathtools}
\usepackage{latexsym,amscd,enumitem}
\usepackage{amsbsy}
\usepackage{color}
\usepackage{multirow}
\usepackage{makecell}
\usepackage{float}
\usepackage{cases}
\usepackage{verbatim}

\usepackage[pagewise,mathlines]{lineno}

 \usepackage[numbers,sort&compress]{natbib}
 \usepackage{bm}
 \usepackage{longtable}
\usepackage[hidelinks]{hyperref}
\newtheorem{theorem}{Theorem}
\allowdisplaybreaks

\newtheorem{lem}{Lemma}

\newtheorem{prop}{Proposition}
\newtheorem{definition}{Definition}
\newtheorem{example}{Example}
\newtheorem{remark}{Remark}

\newtheorem{cons}{Construction}
\newtheorem{coro}{Corollary}
\newcommand{\ba}{\mathbf{a}}
\newcommand{\bs}{\mathbf{s}}
\newcommand{\bt}{\mathbf{t}}
\newcommand{\bu}{\mathbf{u}}

\begin{document}

\title{Constructions of Optimal Frequency-Hopping Sequences with Controlled Minimum Gaps}
\author{
		Chen Li
		\thanks{C. Li and X. Zeng are with the Key Laboratory of Intelligent Sensing System and Security, Ministry of Education, and Hubei Key Laboratory of Applied Mathematics, Faculty of Mathematics and Statistics, Hubei University, Wuhan 430062, China (e-mail: lichen2022@stu.hubu.edu.cn; xiangyongzeng@aliyun.com).},
		Chunlei Li
		\thanks{C. Li and D. Li are with the Department of Informatics, University of Bergen, Bergen, N-5020, Norway (e-mail: chunlei.li@uib.no; dian.li@uib.no).},
		Xiangyong Zeng,	
		and
		Dian Li
}
\date{}
\maketitle

\begin{quote}
{\small {\bf Abstract:}
Frequency-hopping sequences (FHSs) with low Hamming correlation and wide gaps significantly contribute to the anti-interference performance in FH communication systems.
This paper investigates FHSs with optimal Hamming correlation and controlled minimum gaps.
We start with the discussion of the upper bounds on the minimum gaps of uniform FHSs and then
propose a general construction of optimal uniform wide-gap FHSs with length $2l$ and $3l$, which includes the work by Li et al. in IEEE Trans. Inf. Theory, vol. 68, no. 1, 2022 as a special case.
Furthermore, we present a recursive construction of FHSs with length $2l$,
which concatenate shorter sequences of known minimum gaps. It is shown that
the resulting FHSs have the same Hamming correlation as the concatenation-ordering sequences.
As applications,  several known optimal FHSs are used to produce optimal FHSs with controlled minimum gaps.
}

\end{quote}

\section{Introduction} \label{intro}

Frequency-hopping multiple-access (FHMA) systems are extensively employed in modern communication technologies, including ultra-wideband, radar, and Bluetooth \cite{MJ94}. In FH communication systems, mutual interference arises when the same frequencies are used simultaneously. Frequency-hopping sequences (FHSs) are usually finite-length sequences, where each value corresponds to an available frequency in a given time slot. In 1974, Lempel and Greenberger \cite{AH74} introduced the periodic Hamming correlation to assess the anti-interference capabilities of FHSs. From a system security perspective, it is desirable that the sequences exhibit low Hamming correlation, namely, the Hamming cross-correlation between different FHSs and the non-trivial auto-correlation of each FHS should be low.
Two generic lower bounds due to Lempel and Greenberger \cite{AH74} and due to Peng and Fan \cite{DP04} were proposed in the literature.
The past decades witnessed several algebraic or combinatorial constructions of optimal FHSs with respect to the Lempel-Greenberger bound, e.g., \cite{RY04, WC05, CM07, GY09, JK09, JK10, XH12, XH13},
and optimal FHS sets with respect to the Peng-Fan bound, e.g., \cite{WC05, CM07, CJ08, CR09, GY09, XH12, XH13}.

	The rapid development of modern communication systems has challenged the performance of traditional FH communication. Increasingly strained spectrum resources, continuously evolving and complex interference techniques, and intricate, dynamic communication link environments now pose significant limitations.
To address these challenges, FHSs with large or controlled minimum gaps between adjacent frequencies,  referred to as wide-gap FHSs (WGFHSs), have become highly desirable \cite{WZ02, Chen83, Mei94, HY09}. These wide gaps enhance resilience, as they provide effective avoidance of interference and minimize its impact on adjacent frequency channels, even when it occurs.
Consequently, WGFHSs exhibit strong performance against single-frequency narrow-band interference, partial band jamming, tracking interference, and multipath fading. This has driven substantial research interest in WGFHS spread spectrum systems \cite{PCY19, QH22, LC22, PC22}.
The lower bounds on the Hamming correlation of WGFHSs were further investigated \cite{PCY19};
the WGFHSs were generalized to wide-gap-zone FHSs; and the optimality of such WGFHSs and their construction were also studied \cite{QH22}.
Additionally, optimal WGFHSs have been examined in the low-hopping zone \cite{LC22}.
Particularly, in a recent work \cite{PC22}, the authors introduced the construction of two families of optimal WGFHSs over $\mathbb{Z}_l = \{0,1,\dots, l-1\}$ with controlled minimum gap $d-1$,
which concatenated two or three decimated sequences of the trivial sequence $(0,1,\dots, l-1)$ for certain positive integers $l$ and $d$.

This paper aims to advance the study of FHSs with desirable Hamming correlation and minimum gaps.
We first consider the upper bounds on the minimum gaps of uniform FHSs, and then propose two constructions of optimal FHSs. The first construction generalises the work of \cite{PC22} from tuples $(d,d+1)$, $(d,d+1,d+2)$ to $(d_1, d_2)$, $(d_1, d_2, d_3)$, respectively, with $d_1, d_2, d_3$ derived from the difference unit set of $\mathbb{Z}_l$.
In the second construction, we loosen the condition in Theorem \ref{cons I} and consider the cases where $\gcd(l, d_1) = \gcd(l, d_2) = \gcd(l, d_2-d_1) = m \geq 2$.
In this construction two $m\times l_1$ matrices, where $l_1=l/m$, are obtained from the sequence $(0,1,\dots, l-1)$  indexed by $d_1, d_2$, respectively; and  the $2m$ rows of the two matrices are concatenated, in a special order determined by a permutation $\bm{\pi}$ on $\{0,1,\dots, 2m-1\}$, into a sequence $\bu$ of length $2l$.
We show that $\bu$ has the same Hamming correlation as the concatenation-ordering sequence $\bm{\pi}_m$, where $\bm{\pi}_m=(\pi(i)\mod{m})_{0\leq i<2m}$.
As a result, in order to construct an FHS $\bu$ of length $2l$ with optimal Hamming correlation and a controlled minimum gap, one can use an optimal FHS $\bm{\pi}_m$ of length $2m$ as a concatenation-ordering sequence, choose proper integers $d_1$ and $d_2$ to obtain $2m$ subsequences with controlled minimum gaps, which will be concatenated into a sequence of length $2l$.
Furthermore, by appropriately selecting subsequences and sequence sets, this construction method can also be used to construct optimal no-hit-zone multi-timeslot WGFHS sets \cite{XC25}.
Finally, we demonstrate the second construction by considering the optimal FHSs from \cite{JK09, CJ08, JK10}, and compare the sequences constructed in \cite{JK09, CJ08, JK10} with those derived from Theorem \ref{H=H}.

The remainder of this paper is organized as follows. Section 2 introduces some preliminaries about FHSs and their (periodic) Hamming correlation. In Section 3, new upper bounds on the minimum gaps of uniform FHSs and new constructions of optimal uniform WGFHSs are presented. In addition, the applications of using the recursive construction are discussed.
Section 4 concludes the study.

\section{Preliminaries}

	This section introduces some basic definitions and auxiliary results about FHSs.

Let $F=\{f_0, f_1, \dots, f_{l-1}\}$ be a set of $l$ available frequencies.
Let $F^n$ be the set of frequency-hopping sequences (FHSs) of length $n$ over $F$.
For any two FHSs $\mathbf{s}=(s_0, s_1, \dots, s_{n-1})$ and $\mathbf{t}=(t_0, t_1, \dots, t_{n-1})$ in $F^n$, their \emph{(periodic) Hamming cross-correlation} at time delay $\tau$ is defined by
$$
H_{\mathbf{s},\mathbf{t}}(\tau)
=\sum\limits_{i=0}^{n-1}h[s_i, t_{i+\tau}], \,\,\,0\leq \tau< n,
$$
where $h[a,b]=1$ if $a=b$, $h[a,b]=0$ otherwise, and the subscripts are taken modulo $n$.
When $\mathbf{s}=\mathbf{t}$, we call $H_{\mathbf{s},\mathbf{s}}(\tau)$ the \emph{Hamming autocorrelation} of $\mathbf{s}$ and denote it by $H_{\mathbf{s}}(\tau)$.
The Hamming cross-correlation $H_{\mathbf{s},\mathbf{t}}(\tau)$ is equal to $n-{\rm d}_H(\bs, L^{\tau}(\bt))$,
where ${\rm d}_H(\bs, L^{\tau}(\bt))$ denotes the Hamming distance between the sequences $\bs$ and $L^{\tau}(\bt)$, which is the left cyclic shift of $\bt$ by $\tau$ positions.
 Hence there have been some discussions on FHSs in the context of coding theory \cite{CR09, XC20}. For two sequences $\mathbf{s}$, $\mathbf{t}\in F^n$, let
$$
H_{\mathbf{s}}=\max\limits_{0< \tau<n}H_{\mathbf{s}}(\tau) \quad\text{and} \quad
H_{\mathbf{s},\mathbf{t}}=\max\limits_{0\leq \tau<n}H_{\mathbf{s},\mathbf{t}}(\tau).
$$
In 1974, Lempel and Greenberger \cite{AH74} gave the criteria of optimality of FHSs:
 if $H_{\mathbf{s}}\leq H_{\mathbf{s}'}$ for all $\mathbf{s}'\in F^n$, then $\mathbf{s}$ is called an \emph{optimal FHS}.
They also established a lower bound on $H_{\mathbf{s}}$ as follows.
\begin{lem}(Lempel-Greenberger bound \cite{AH74})\label{LG}
For every FHS $\mathbf{s}$ of length $n$ over an alphabet $F$ of size $\ell$, we have
\begin{equation}\label{lg-1}
H_{\mathbf{s}}\geq \left\lceil\frac{(n-\varepsilon)(n+\varepsilon-\ell)}{\ell(n-1)}\right\rceil,
\end{equation}
where $\varepsilon$ is the least nonnegative residue of $n$ modulo $\ell$ and $\lceil x\rceil$ represents the ceiling function of a real number $x$.
\end{lem}
Assume $n = kl + \varepsilon$ with $0\leq \varepsilon<l$. It can be verified that the Lempel-Greenberger bound
becomes \cite{RY04}
\[
H_{\bs} \geq
\begin{cases}
k, & \text{ if }\, n > l, \\
0, & \text{ if }\, n \leq l.
\end{cases}
\]
 An FHS is said to be \emph{optimal} if it achieves the Lempel-Greenberger bound  \cite{AH74}.
In the derivation of the Lempel-Greenberger bound, it can be seen that an optimal FHS is required to take frequencies as uniformly as possible. We recall \textit{uniform sequences} below.

\begin{definition}(\cite{XC10})
	Let $\mathbf{s}=(s_0, s_1, \dots, s_{n-1})$ be an FHS over $F=\{f_0,f_1,\dots, f_{\ell-1}\}$, and $N_{f_k}(\bs)$ the number of $f_k$ in the sequence $\bs$.
	Then $\bs$ is said to be uniform if
	\[
	\max_{0\leq k<\ell}N_{f_k}(\bs) -\min_{0\leq k<\ell}N_{f_k}(\bs) = \begin{cases}
	0, & \text{if} \,\, \, \ell \mid n,
	\\
	1, & \text{otherwise}.
	\end{cases}
	\]
\end{definition}

In conventional applications, such as FH-CDMA systems, FHSs are required to have low Hamming correlation.
In addition, the minimum gaps of FHSs have gained more attention recently \cite{WF23, PCY19, QH22, LC22, PC22}. We recall this notion below \cite{WZ02, Chen83, Mei94}.

\begin{definition}
The minimum gap of an FHS  $\mathbf{s}=(s_0, s_1, \dots, s_{n-1}) $ over $F$ is defined as
$$
g = \min_{0\leq i<n-1}\{|s_{i+1}-s_i|, |s_{n-1}-s_0|\}-1.
$$
\end{definition}
By taking into consideration the minimum gap, the authors of \cite{PCY19} improved the Lempel-Greenberger bound as follows.
\begin{lem}(\cite{PCY19})\label{WGLG}
For a WGFHS $\mathbf{s}$ of length $n$ over an alphabet $F$ of size $\ell$, we have
\begin{equation}\label{lg-3}
H_{\mathbf{s}}\geq \left\lceil\frac{(n-\varepsilon)(n+\varepsilon-\ell)}{\ell(n-3)}\right\rceil,
\end{equation}
where $\varepsilon$ is the least nonnegative residue of $n$ modulo $\ell$.
\end{lem}

Notice that the Hamming correlation between two sequences in $F^n$ is determined by the number of the same components in the two sequences.
In the context of Hamming correlation, the alphabet $F=\{f_0,f_1,\dots, f_{l-1}\}$ can be treated as $F=\{0,1,\dots, l-1\}$.
The simplified alphabet $F=\{0,1,\dots, l-1\}$ can, to some extent, reflect the minimum gap of sequences in $F^n$.
For a frequency alphabet $F$ of size $l$, one can consider the sorted version $F' = \{f'_0,f'_1,\dots, f'_{l-1}\}$ with $f'_0<f'_1<\dots<f'_{l-1}$  and $f' = \min \{ |f'_{i+1}-f'_{i}|: 0\leq i \leq l-1 \}$, where the subscripts are taken modulo $l$.
Given a sequence $\bs=(s_0,s_1,\dots, s_{n-1})$ over $\{0,1,\dots, l-1\}$ with minimum gap $g$,
one may define $\bs'=(f'_{s_0}, f'_{s_1},  \dots, f'_{s_{n-1}})$.
It is easily seen that
the minimum gap of the sequence $\bs'$ is at least $(g+1)\cdot f'-1$.
Following the conventional notation in the literature \cite{AH74,DP04,PC22}, this paper denotes the frequency alphabet $F$ of size $l$ as $\{0,1,\dots,l-1\}$ and use the following notation:
\begin{itemize}[itemsep=-3pt]
	\item  the frequency alphabet $F$ is assumed as the residue ring $\mathbb{Z}_l$;
	\item  $\mathbf{a} \Vert \mathbf{b}$ denotes the concatenation of sequences $\mathbf{a}$ and $\mathbf{b}$;
	\item $(n, l, \lambda)$ denotes an FHS of length $n$ over $\mathbb{Z}_l$ having Hamming correlation $\lambda$;
    \item $(c)_d$ denotes the least nonnegative residue of $c$ modulo $d$;
    \item $\bm{\pi}$ denotes a permutation on $\{0,1,\dots, 2m-1\}$.
\end{itemize}

\section{Constructions of Optimal WGFHSs}

This section starts with the discussion on minimum gaps of FHSs, and then presents two general constructions of FHSs that achieve optimal Hamming correlation and have controlled minimum gaps.

\subsection{Minimum gaps of FHSs}

Note that for a uniform sequence $\bs=(s_0,s_1,\dots, s_{n-1})$ over $\mathbb{Z}_l$ with $n=kl$, the lower bound in Lemma \ref{WGLG}
is the same as that in Lemma \ref{LG} since both reduce to $H_{\mathbf{s}}\geq k.$
As a matter of fact, for a WGFHS, its minimum gap does not play a role in the lower bound in Lemma \ref{WGLG},
which appears insufficient to assess the optimality of WGFHSs when the minimum gap needs to be considered.
Such insufficiency seems to stem from the independence of the two metrics:
 the Hamming correlation is only concerned with the number of the same entries in sequences while the minimum gap inspects the actual gaps between consecutive entries.
For better assessing the optimality of WGFHSs, we study upper bounds on the minimum gaps of uniform FHSs below.

\begin{prop}\label{Prop1}
Let $\mathbf{s}=(s_0, s_1, \dots, s_{n-1})$ be a uniform FHS over $\mathbb{Z}_\ell$ with minimum gap $g$. Then
\[
g\leq
\begin{cases}
\frac{\ell}{2}-1, & \text{if $\ell \nmid n$ and $2 \mid \gcd(\ell, n)$},\\
\lfloor\frac{\ell-1}{2}\rfloor-1, & \text{otherwise},
\end{cases}
\]
where $\lfloor x \rfloor$ represents the floor function of a real number $x$. Moreover, the above upper bounds are reachable.
\end{prop}
\begin{proof}
When $l$ is odd, $g<|s_i-\lfloor\frac{l}{2}\rfloor|\leq \lfloor\frac{l}{2}\rfloor$, i.e. $g \leq \lfloor\frac{l-1}{2}\rfloor-1$.
Furthermore, there exist sequences with minimum gap $g$ reaching $\lfloor\frac{l-1}{2}\rfloor-1$.
Take $\mathbf{s}=(0, d, \dots, (n-1)d)$ with $d = \frac{l-1}{2}$ over $\mathbb{Z}_l$.
Since $\gcd(l, l-1) = \gcd(l, \frac{l-1}{2}) = 1$, the sequence $\mathbf{s}$ is uniform. For $0\leq i \leq n-2$, $|s_{i+1}-s_i| = d$ and $|s_{n-1}-s_0| = (n-1)d\equiv l- d \mod l$.
This implies $g = \frac{l-1}{2}-1$,
thereby, the upper bound $\lfloor\frac{l-1}{2}\rfloor-1$ is tight when $l$ is odd.

When $l$ is even, assume $g = \frac{l}{2}-1$, the frequencies $0$ and $\frac{l}{2}$ can only appear in the sequence in the form of $(0, \frac{l}{2}, 0, \frac{l}{2}, \dots, 0, \frac{l}{2}, 0)$. Clearly, such sequence is not uniform for $l \mid n$. Hence the minimum gap satisfies $g<\frac{l}{2}-1$. Below we shall show that there exist uniform sequences with minimum gap $g$ reaching $\frac{l}{2}-2$.

Let $d = \frac{l}{2}-1$, $\gcd(l,d) = m$ and $n_1 = \frac{n}{m}$. Define a sequence
$$\mathbf{s} = (0, d, \dots, (n_1-1)d, 1, d+1, \dots, (n_1-1)d+1, \dots, m-1, d+m-1, \dots, (n_1-1)d+m-1)$$
over $\mathbb{Z}_l$. We have
\[
\begin{split}
& \min\{|s_{i+1}-s_i|, |s_{n-1}-s_0|: 0\leq i \leq n-2\}
\\= &\min\{d, (n_1-1)d-1, (n_1-1)d+m-1\}
\\ = &\min\{d, l-d-1, l-d+m-1\}
\\ = & \,d,
\end{split}
\]
where $(n_1-1)d = \frac{d}{m}\cdot n - d \equiv l-d \mod l$.
This implies that the upper bound $\frac{l}{2}-2 = \lfloor\frac{l-1}{2}\rfloor-1$ is reachable when  $l \mid n$ and $l$ is even.

Now we only need to discuss the case $l\nmid n$ for even $l$.
Clearly $g \leq \frac{l}{2}-1$ since $|s_{i}-\frac{l}{2}|\leq \frac{l}{2}$ for any $s_i\in \mathbb{Z}_l$ and $0\leq i \leq n-1$.
 Below, we shall show that $\frac{l}{2}-1$ is the tight upper bound on $g$ when $n$ is even and $\frac{l}{2}-2$ is the tight upper bound on $g$ when $n$ is odd by providing explicit sequences.

$\mathbf{Case\,\, 1}$. For even integers $n$, $l$ and $l \nmid n$, there exist a positive integer $q$ and an even number $2\leq r\leq \frac{l}{2}$ such that $n=q l\pm r$.
Since the discussions of $n=q l+r$ and $n=q l-r$ are similar, we will only focus on $n=q l+r$.
Take $d=\frac{l}{2}$ and define $d$ sequences as follows:
	$$
	\begin{aligned}
	\bu^j &= (j, d+j, \dots, j, d+j, j), \quad j\in \{1\} \cup \{0, 2, 4, \dots, r-4\},\\
	\mathbf{v}^k &= (d+k, k, \dots, d+k, k, d+k), \quad k\in \{ r-2, d-1\} \cup \{3, 5, 7, \dots, r-3\}, \\
	\mathbf{w}^t &= (d+t, t, \dots, d+t, t), \quad t\in \{r-1, r, \dots, d-2\}, \\
	\end{aligned}
	$$
where the sequences $\bu^j$ and $\mathbf{v}^k$ have length $2q+1$ and $\mathbf{w}^t$ has length $2q$.
In particular, if $r = 2$, then $j = 0$ and $k = d-1$. According to the value of $r$, we define sequences $\mathbf{s}$ as follows:
$$
	\begin{aligned}
    \mathbf{s} &= \bu^0 \Vert \mathbf{w}^{r-1} \Vert \mathbf{w}^{r} \Vert \cdots \Vert \mathbf{w}^{d-2} \Vert \mathbf{v}^{d-1} \,\, \text{for} \,\, r=2;\\
    \mathbf{s} &= \bu^0 \Vert \mathbf{v}^{r-2} \Vert \mathbf{u}^{1} \Vert \mathbf{w}^{r-1} \Vert \mathbf{w}^{r} \Vert \cdots \Vert \mathbf{w}^{d-2} \Vert \mathbf{v}^{d-1} \,\, \text{for} \,\, r=4;\\
    \mathbf{s} &= \bu^0 \Vert \mathbf{v}^{r-2} \Vert \mathbf{u}^{r-4} \Vert \mathbf{v}^{r-3} \Vert \mathbf{u}^{r-6} \Vert \mathbf{v}^{r-5} \Vert \cdots \Vert \mathbf{u}^{2} \Vert \mathbf{v}^{3} \Vert \mathbf{u}^{1} \Vert \mathbf{w}^{r-1} \Vert \mathbf{w}^{r} \Vert \cdots \Vert \mathbf{w}^{d-2} \Vert \mathbf{v}^{d-1} \,\, \text{for} \,\, r>4.
\end{aligned}
	$$
For $j$, $k$ and $t$ as specified above, we know that $j$ and $d+k$ occur exactly $q+1$ times in the sequence $\mathbf{s}$, and $d+j$, $k$, $t$ and $d+t$ occur exactly $q$ times. This means that $\mathbf{s}$ is uniform. It's easy to verify that $\min\{|s_{n-1}-s_0|, |s_{i+1}-s_i|: 0\leq i \leq n-2\} = \frac{l}{2}$, i.e., $g = \frac{l}{2}-1$.
Hence $\frac{l}{2}-1$ is the tight upper bound on $g$ when $l \nmid n$ and $2 \mid \gcd(l, n)$.

$\mathbf{Case\,\, 2}$.
       Assuming there exists a sequence $\mathbf{s}=(s_0, s_1, \dots, s_{n-1})$ with $g=\frac{l}{2}-1$ over $\mathbb{Z}_l$, where $n$ is odd and $l$ is even.
Define two sets $A = \{0, 1, \dots, \frac{l}{2}-1\}$ and $B = \{\frac{l}{2}, \frac{l}{2}+1, \dots, l-1\}$.
To ensure $g=\frac{l}{2}-1$, adjacent entries in the sequence $\mathbf{s}$ cannot belong to the same set. Without loss of generality, let $s_0 \in A$ and $s_1 \in B$. Since $n$ is odd, we have $s_{n-1} \in A$. Thus, $|s_{n-1}-s_0|\leq \frac{l}{2}-1$, which leads to a contradiction.
Below, we demonstrate the existence of uniform sequences with $g=\frac{l}{2}-2$.

Let $\mathbf{s} = (s_0, s_1, \dots, s_{n-2})$ be a uniform FHS as defined in Case 1.
Inserting $\frac{l}{2}$ or $\frac{l}{2}-1$ at the beginning or the end of the sequence $\mathbf{s}$ does not change its uniformity, and the minimum gap of the new sequence is $\frac{l}{2}-2$.
Hence $\lfloor\frac{l-1}{2}\rfloor-1$ is the tight upper bound of $g$ when $l$ is even,  $n$ is odd and $l \nmid n$.
\end{proof}

In the following subsections, we will propose two constructions of FHSs of which the minimum gaps can be easily controlled.

\subsection{The first construction}

This subsection presents a general construction of optimal WGFHSs which include the following FHSs in \cite{PC22} as special cases.

\begin{prop} (\cite{PC22}) \label{Prop_PC22_1}
	 Let $\ell$ and $d$ be two positive integers with $1<d<$ $\frac{\ell}{2}$ and $\operatorname{gcd}(\ell, d)=\operatorname{gcd}(\ell, d+1)=1$.
	Define two sequences $\mathbf{s}=\left(s_0, s_1, \dots, s_{\ell-1}\right)$ and $\mathbf{t}=\left(t_0, t_1, \dots, t_{\ell-1}\right)$ as follows:
	$$
	\begin{aligned}
	s_i & =(i \cdot d)_{\ell}, \quad 0 \leq i<\ell, \\
	t_i & =(i \cdot(d+1)+1)_{\ell}, \quad 0 \leq i<\ell.
	\end{aligned}
	$$
	 Then the sequence
	$\mathbf{v} = \bs \Vert \bt$ is an optimal $(2\ell, \ell, 2)$-FHS with minimum gap $d-1$.
\end{prop}
\begin{prop} (\cite{PC22})\label{Prop_PC22_2}
	Let $\ell$ and $d$ be two positive integers with $1<d<$ $\frac{\ell-1}{2}$ and $\operatorname{gcd}(\ell, d)=\operatorname{gcd}(\ell, d+1)=\gcd(\ell, d+2)=1$.
	Define three sequences $\mathbf{s}=\left(s_0, s_1, \dots, s_{\ell-1}\right)$, $\mathbf{t}=\left(t_0, t_1, \dots, t_{\ell-1}\right)$ and $\mathbf{u}=\left(u_0, u_1, \dots, u_{\ell-1}\right)$ as follows:
	$$
	\begin{aligned}
	s_i & =(i \cdot d )_{\ell}, \quad 0 \leq i<\ell, \\
	t_i & =(i \cdot(d+1)+1)_{\ell}, \quad 0 \leq i<\ell, \\
	u_i & =(i \cdot (d+2) +2)_{\ell}, \quad 0 \leq i<\ell. \\
	\end{aligned}
	$$
	 Then the sequence
	$\mathbf{v} = \bs \Vert \bt \Vert \mathbf{u} $ is an optimal $(3\ell, \ell, 3)$-FHS with minimum gap $d-1$.
\end{prop}

Given a positive integer $l$, the \textit{difference unit set} of $\mathbb{Z}_l$, denoted as $DU(\mathbb{Z}_l)$,
is a subset $D$ of $\mathbb{Z}_l^*=\{1\leq x\leq l-1\,:\, \gcd(x,l) = 1 \}$ such that for any integers $d_1,\, d_2\in D$ with $d_1<d_2$,
$d_2-d_1\in \mathbb{Z}_l^*$ and $D$ is the maximal set satisfying the aforementioned property.
Suppose that the factorization of $l$ is given by $l = p_1^{e_1}p_2^{e_2}\dots p_r^{e_r}$, where $r\geq 1$, $e_i\geq1$ for $1\leq i \leq r$, and $p_i$s are distinct primes such that $2\leq p_1<p_2<\cdots<p_r$. Then $DU(\mathbb{Z}_l)$ has size $p_1-1$, and one
example is $DU(\mathbb{Z}_l) = \{1,2,\dots, p_1-1\}$ \cite{JK09}.
It is readily seen that in Propositions \ref{Prop_PC22_1} and \ref{Prop_PC22_2},
the integers $d, d+1$ and $d+2$ belong to the different unit set of $\mathbb{Z}_l$. More importantly, it was shown in \cite{Cao2006Nov}
that FHSs from $DU(\mathbb{Z}_l)$ satisfy the following property.

\begin{prop} (\cite{Cao2006Nov})\label{Prop_Cao_3}
	Let $DU(\mathbb{Z}_\ell)$ be a difference unit set of $\mathbb{Z}_\ell$. For $d \in DU(\mathbb{Z}_\ell)$,
	define a sequence $\bs^{d} = (s^d_i)_{i=0}^{\ell-1}$ with $s^d_i = (i\cdot d)_\ell$.
	Then	
	$$
	H_{\bs^{d_1},\, \bs^{d_2}}(\tau)=
	\begin{cases}
	  \ell, & \text { if } d_1=d_2 \text { and } \tau=0, \\
	  0, & \text { if } d_1=d_2 \text { and } \tau \neq 0, \\
	 1, & \text { if } d_1 \neq d_2.
	 \end{cases}
	$$
\end{prop}

Based on the above property, a general construction of optimal WGFHSs is provided below.

\begin{theorem}\label{cons I}
	Let $DU(\mathbb{Z}_\ell)$ be a difference unit set of $\mathbb{Z}_\ell$. Let $d_1$, $d_2$ and $d_3$ be three different elements in $DU(\mathbb{Z}_\ell)$,
	and $\bs^{d}$ be a sequence of length $\ell$ given by  $s^{d}_i = (i \cdot d)_\ell$ for $d\in\{d_1,d_2,d_3\}$.
	Then,
	\begin{itemize}
		\item the sequence $\mathbf{v} = \bs^{d_1} \Vert \bs^{d_2}$ is an optimal $(2\ell, \ell, 2)$-FHS with minimum gap $g_1$;
		\item the sequence $\mathbf{v} = \bs^{d_1} \Vert \bs^{d_2}\Vert \bs^{d_3}$ is an optimal $(3\ell, \ell, 3)$-FHS with minimum gap $g_2$,
	\end{itemize}
where $g_1 = \min\limits_{1\leq j \leq 2} \{d_j-1, \ell-d_j-1\}$ and $g_2 = \min\limits_{1\leq j \leq 3} \{d_j-1, \ell-d_j-1\}$.	
\end{theorem}
\begin{proof}

The two cases will be proved in a similar manner. We therefore omit the proof for the first case and focus on the second one, which is more complicated.

For $\mathbf{v} = \bs^{d_1} \Vert \bs^{d_2}\Vert \bs^{d_3}$, since $H_{\mathbf{v}}(\tau) = H_{\mathbf{v}}(3l-\tau)$ for $0< \tau < 3l$,
 it suffices to calculate $H_{\mathbf{v}}(\tau)$ for $0< \tau \leq \frac{3l}{2}$.
For convenience, we consider $H_{\mathbf{v}}(\tau)$ for $0< \tau < 2l$.

When $0 < \tau < l$, from Proposition \ref{Prop_Cao_3}, we have
\begin{align*}
H_{\mathbf{v}}(\tau)
&=\sum\limits_{i=0}^{l-1-\tau}h\left[s^{d_1}_i,s^{d_1}_{i+\tau}\right]+\sum\limits_{i=l-\tau}^{l-1}h\left[s^{d_1}_i,s^{d_2}_{i+\tau-l}\right]+
\sum\limits_{i=0}^{l-1-\tau}h\left[s^{d_2}_i,s^{d_2}_{i+\tau}\right]+\sum\limits_{i=l-\tau}^{l-1}h\left[s^{d_2}_i,s^{d_3}_{i+\tau-l}\right]\\
&\,\,\,\,\,\,\,\,+\sum\limits_{i=0}^{l-1-\tau}h\left[s^{d_3}_i,s^{d_3}_{i+\tau}\right]
+\sum\limits_{i=l-\tau}^{l-1}h\left[s^{d_3}_i,s^{d_1}_{i+\tau-l}\right]\\
&=\sum\limits_{i=l-\tau}^{l-1}h\left[s^{d_1}_i,s^{d_2}_{i+\tau-l}\right]+\sum\limits_{i=l-\tau}^{l-1}h\left[s^{d_2}_i,s^{d_3}_{i+\tau-l}\right]
+\sum\limits_{i=l-\tau}^{l-1}h\left[s^{d_3}_i,s^{d_1}_{i+\tau-l}\right]\\
&\leq H_{\bs^{d_1},\, \bs^{d_2}} + H_{\bs^{d_2},\, \bs^{d_3}} + H_{\bs^{d_3},\, \bs^{d_1}}\\
&= 3.
\end{align*}

When $\tau = l$, from Proposition \ref{Prop_Cao_3}, we have
\begin{align*}
H_{\mathbf{v}}(\tau)
&=\sum\limits_{i=0}^{l-1}h\left[s^{d_1}_i,s^{d_2}_i\right]+\sum\limits_{i=0}^{l-1}h\left[s^{d_2}_i,s^{d_3}_i\right]
+\sum\limits_{i=0}^{l-1}h\left[s^{d_3}_i,s^{d_1}_i\right]\\
&=H_{\bs^{d_1},\, \bs^{d_2}} + H_{\bs^{d_2},\, \bs^{d_3}} + H_{\bs^{d_3},\, \bs^{d_1}}\\
&= 3.
\end{align*}

When $l < \tau < 2l$, let $\tau' = \tau - l$, then $0 < \tau' < l$ and
\begin{align*}
H_{\mathbf{v}}(\tau)
&=\sum\limits_{i=0}^{l-1-\tau'}h\left[s^{d_1}_i,s^{d_2}_{i+\tau'}\right]+\sum\limits_{i=l-\tau'}^{l-1}h\left[s^{d_1}_i,s^{d_3}_{i+\tau'-l}\right]+
\sum\limits_{i=0}^{l-1-\tau'}h\left[s^{d_2}_i,s^{d_3}_{i+\tau'}\right]\\
&\,\,\,\,\,\,\,\,+\sum\limits_{i=l-\tau'}^{l-1}h\left[s^{d_2}_i,s^{d_1}_{i+\tau'-l}\right]+\sum\limits_{i=0}^{l-1-\tau'}h\left[s^{d_3}_i,s^{d_1}_{i+\tau'}\right]+\sum\limits_{i=l-\tau'}^{l-1}h\left[s^{d_3}_i,s^{d_2}_{i+\tau'-l}\right].
\end{align*}
To prove $H_{\mathbf{v}}(\tau) \leq 3$, we will prove that
\begin{align}\label{d12}
\sum\limits_{i=0}^{l-1-\tau'}h\left[s^{d_1}_i,s^{d_2}_{i+\tau'}\right]+
\sum\limits_{i=l-\tau'}^{l-1}h\left[s^{d_2}_i,s^{d_1}_{i+\tau'-l}\right]
\leq 1,
\end{align}
\begin{align}\label{d13}
\sum\limits_{i=l-\tau'}^{l-1}h\left[s^{d_1}_i,s^{d_3}_{i+\tau'-l}\right]+
\sum\limits_{i=0}^{l-1-\tau'}h\left[s^{d_3}_i,s^{d_1}_{i+\tau'}\right]
\leq 1
\end{align}
and
\begin{align}\label{d23}
\sum\limits_{i=0}^{l-1-\tau'}h\left[s^{d_2}_i,s^{d_3}_{i+\tau'}\right]
+\sum\limits_{i=l-\tau'}^{l-1}h\left[s^{d_3}_i,s^{d_2}_{i+\tau'-l}\right]
\leq 1.
\end{align}
It follows from Proposition \ref{Prop_Cao_3} that $H_{\bs^{d_1},\, \bs^{d_2}} = 1$ for $d_1 \neq d_2$.
Assume (\ref{d12}) does not hold. Then we have
$$\sum\limits_{i=0}^{l-1-\tau'}h\left[s^{d_1}_i,s^{d_2}_{i+\tau'}\right]=\sum\limits_{i=l-\tau'}^{l-1}h\left[s^{d_2}_i,s^{d_1}_{i+\tau'-l}\right]=1,$$
i.e.,
\begin{align*}
\sum\limits_{i=0}^{l-1-\tau'}h\left[s^{d_1}_i,s^{d_2}_{i+\tau'}\right]
=\sum\limits_{i=0}^{l-1-\tau'}h\left[id_1,(i+\tau')d_2\right]
=\sum\limits_{i=0}^{l-1-\tau'}h\left[i,(d_1-d_2)^{-1}\tau' d_2\right]
=1
\end{align*}
and
\begin{align*}
\sum\limits_{i=l-\tau'}^{l-1}h\left[s^{d_2}_i,s^{d_1}_{i+\tau'-l}\right]
=\sum\limits_{i=l-\tau'}^{l-1}h\left[id_2,(i+\tau'-l)d_1\right]
=\sum\limits_{i=l-\tau'}^{l-1}h\left[i,(d_2-d_1)^{-1}(\tau'-l)d_1\right]
=1.
\end{align*}
Denote $h_1=(d_1-d_2)^{-1}\tau' d_2$ and $h_2=(d_2-d_1)^{-1}(\tau'-l)d_1$. Then
$h_1 \in [0, l-1-\tau']$ and $h_2 \in [l-\tau', l-1]$, which implies $h_1 + h_2 \in [l-\tau', 2l-2-\tau']$.
 Furthermore,
\begin{align}\label{h_1+h_2}
h_1+h_2
&=(d_1-d_2)^{-1}\tau' d_2+(d_2-d_1)^{-1}(\tau'-l)d_1 \nonumber\\
&=\tau'(d_1-d_2)^{-1}(d_2-d_1)-(d_2-d_1)^{-1}l d_1 \nonumber\\
& \equiv l-\tau' \mod l,
\end{align}
which implies $h_1 = 0$ and $h_2 = l-\tau'$, thereby $l \mid (d_1-d_2)^{-1}\tau' d_2$.
Since $d_1, d_2 \in DU(\mathbb{Z}_l)$, we get $l \mid \tau'$ . This result contradicts $0 < \tau' < l$.

Thus the inequality (\ref{d12}) holds. The inequalities (\ref{d13}) and (\ref{d23}) can be similarly proved. Consequently, we have $H_{\mathbf{v}}(\tau) \leq 3$.

For $\mathbf{v} = (v_i)^{3l-1}_{i=0} = \bs^{d_1} \Vert \bs^{d_2} \Vert \bs^{d_3}$, it is easily seen that $|v_{i+1}-v_i| \in \{d_1, d_2, d_3, (l-1)d_1, (l-1)d_2, (l-1)d_3\}$, where the subscripts of $v$ are taken modulo $3l$.
Then the minimum gap of $\mathbf{v}$ is $g_2 = \min\limits_{1\leq j\leq 3} \{d_j, (l-1)d_j\} -1 = \min\limits_{1\leq j\leq 3} \{d_j-1, l-d_j-1\}$, where $(l-1)d_j \equiv l-d_j \mod l$.
\end{proof}

The following example demonstrates Theorem \ref{cons I}.

\begin{example}
Let $\ell=25$.
According to the definition of the difference unit set, $DU(\mathbb{Z}_{25})$ is not unique. Since their applications in Theorem \ref{cons I} are similar, we only consider the case $DU(\mathbb{Z}_{25}) = \{3, 6, 7, 9\}$.

When $d_1 = 7$ and $d_2 = 9$, $\mathbf{v} = \bs^{d_1} \Vert \bs^{d_2} = (0, 7, 14, 21, 3, 10, 17, 24, 6, 13, 20, 2, 9, 16, 23, 5,\\
 12, 19, 1, 8, 15, 22, 4, 11, 18, 0, 9, 18, 2, 11, 20, 4, 13, 22, 6, 15, 24, 8, 17, 1, 10, 19, 3, 12, 21, 5, 14,\\
  23, 7, 16)$.
Then $\mathbf{v}$ is a sequence of length 50 over $\mathbb{Z}_{25}$.
By calculation, we get $H_{\mathbf{v}} = 2$, and the minimum gap is $d_1-1 = 6$.

When $d_1 = 6$, $d_2 = 7$ and $d_3 = 9$, we have $\mathbf{v} = \bs^{d_1} \Vert \bs^{d_2}\Vert \bs^{d_3} = (0, 6, 12, 18, 24, 5, 11, 17, 23,\\
 4, 10, 16, 22, 3, 9, 15, 21, 2, 8, 14, 20, 1, 7, 13, 19, 0, 7, 14, 21, 3, 10, 17, 24, 6, 13, 20, 2, 9, 16, 23, 5,\\
 12, 19, 1, 8, 15, 22, 4,  11, 18, 0, 9, 18, 2, 11, 20, 4, 13, 22, 6, 15, 24, 8, 17, 1, 10, 19, 3, 12, 21, 5, 14,\\
 23, 7, 16)$
and the sequence $\mathbf{v}$ is an optimal $(75, 25, 3)$-FHS with minimum gap $5$.
\end{example}

\begin{remark}
{\rm(i)} The first case of Theorem \ref{cons I} can be further extended.
Let $s^{d_1}_i = (i\cdot d_1+i_1)_\ell$ and $s^{d_2}_i = (i\cdot d_2+i_2)_\ell$, where $i_1$ and $i_2$ are two integers with $0 \leq i_1, i_2 < \ell$, the Hamming correlation of $\mathbf{v} = \bs^{d_1} \Vert \bs^{d_2}$ is still optimal.  When $1<d_1<\frac{\ell}{2}$, $d_2=d_1+1$ and $(i_1,i_2)=(0,1)$, this case corresponds to \cite[Thm. 1]{PC22}.
The parameters in Theorem \ref{cons I} of this paper are more flexible.
For example, $\ell=25$, $d_1=6$ and $d_2=7$, $8$ or $9$.
When $d_1 = \lfloor \frac{\ell-1}{2}\rfloor$ and $d_2 = d_1+1$, the minimum gap can achieve the upper bound stated in Proposition \ref{Prop1}. Otherwise, the minimum gap is always less than the upper bound.

{\rm(ii)} For the second case of Theorem \ref{cons I}, when choosing $(d_1,d_2,d_3) = (d, d+1, d+2)$ with $1<d<\frac{\ell-1}{2}$, one can extend it by shifting each subsequence by $i_1, i_2, i_3$ positions as discussed in \cite[Rem. 2]{PC22}.
For any $d_1, d_2, d_3 \in DU(\mathbb{Z}_\ell)$ and $0 \leq i_1, i_2, i_3 < \ell$,
this kind of extension may not work. The main reason is that now
$h_1=(d_1-d_2)^{-1}\tau' d_2+(d_2-d_1)^{-1}(i_1-i_2)$ and $h_2=(d_2-d_1)^{-1}(\tau'-l)d_1+(d_2-d_1)^{-1}(i_1-i_2)$.
In this case, $h_1 \in [0, \ell-1-\tau']$ and $h_2 \in [\ell-\tau', \ell-1]$ are possible,
indicating that (\ref{d12}) may not hold.
Two of (\ref{d12}), (\ref{d13}) and (\ref{d23}) do not hold, the sequence $\mathbf{v}$ is not optimal.
For example, let $\ell=13$ and $(d_1,d_2,d_3) = (4, 5, 7)$, when $i_1=1$, $i_2=3$ and $i_3=4$, the sequence is not optimal.
For the sequences of length $3\ell$, in reference \cite{PC22} and in Theorem \ref{cons I}, the minimum gaps are always less than the upper bound stated in Proposition \ref{Prop1}.

In addition, one cannot easily choose $k\geq 4$ integers $d_1,\dots, d_k $ in $DU(\mathbb{Z}_\ell)$ to obtain optimal $(k\ell, \ell, k)$-FHSs with controlled minimum gaps.

\end{remark}

\subsection{A recursive construction}\label{recursive}
The previous subsection discussed the case of $\gcd(l, d_i)=1$. Now we will further explore the case that $\gcd(l, d_i)=m$ with $m\geq 2$.

Let $l = l_1m$ for a positive integer $m\geq 2$. We arrange the sequence $(0,1,\dots,l-1)$  as an $m\times l_1$ matrix $A$ as follows:
\begin{align*}
A=[A_0, A_1,\dots, A_{l_1-1}] =
\begin{bmatrix}
 0&m&\cdots&(l_1-1)m\\
1&m+1&\cdots&(l_1-1)m+1\\
\vdots&\vdots&&\vdots\\
m-1&2m-1&\cdots&(l_1-1)m+m-1\\
\end{bmatrix},
\end{align*}
where $A_j$ is the $j$-th column of $A$.
For a permutation $\bm{\rho}$ of $\{0,1,\dots, l_1-1\}$, denote
$$
\bm{\rho}(A)=[A_{\rho(0)}, A_{\rho(1)}, \dots, A_{\rho(l_1-1)}].
$$
Below we will use two particular permutations in the main construction.

Let $d_1$ and $d_2$ be two positive integers with $2\leq d_1 < d_2 < l$ and
\begin{equation}\label{eq_cond_d1d1}
\gcd(l, d_1) = \gcd(l, d_2) = \gcd(l, d_2-d_1) = m.
\end{equation}
Denote $e_1=d_1/m$ and $e_2=d_2/m$. Then $\gcd(l_1,e_1)=\gcd(l_1,e_2)=\gcd(l_1,e_2-e_1)=1$.
Let $\bm{\rho}_1=[0, e_1, \dots, (l_1-1)e_1]$ and $\bm{\rho}_2=[0, e_2, \dots, (l_1-1)e_2]$ be two permutations on $\{0, 1, \dots, l_1-1\}$,
where the entries in $\bm{\rho}_1$ and $\bm{\rho}_2$ are taken modulo $l_1$.
Then we obtain the following two $m\times l_1$ matrices:
\begin{align*}
&S=\bm{\rho}_1(A)=
\begin{bmatrix}
 0&d_1&\cdots&(l_1-1)d_1\\
1&d_1+1&\cdots&(l_1-1)d_1+1\\
\vdots&\vdots&&\vdots\\
m-1&d_1+m-1&\cdots&(l_1-1)d_1+m-1\\
\end{bmatrix}
=
\begin{bmatrix}
\bs^0 \\ \bs^1 \\ \vdots \\ \bs^{m-1}
\end{bmatrix}
\end{align*}
and
\begin{align*}
&T=\bm{\rho}_2(A)=
 \begin{bmatrix}
 0&d_2&\cdots&(l_1-1)d_2\\
 1&d_2+1&\cdots&(l_1-1)d_2+1\\
 \vdots&\vdots&&\vdots\\
 m-1&d_2+m-1&\cdots&(l_1-1)d_2+m-1\\
 \end{bmatrix}
 =
\begin{bmatrix}
\bt^0 \\ \bt^1 \\ \vdots \\ \bt^{m-1}
\end{bmatrix},
\end{align*}
where entries in $S$ and $T$ are taken modulo $l$.
With the above preparation, we present the following construction.

\begin{cons} \label{cons_II}
	Let $\ell$, $d_1$ and $d_2$  be three positive integers with $\gcd(\ell, d_1) = \gcd(\ell, d_2) = \gcd(\ell, d_2-d_1) = m$.
	Let $\mathbf{s}^j$ and $\mathbf{t}^k$ be the row vectors of $S$ and $T$, respectively, given by
	\begin{equation}\label{st}
	\begin{aligned}
	&\mathbf{s}^j=(s^j_0, s^j_1, \dots, s^j_{\ell_1-1}), \,\,\,\,\,\,s^j_i=id_1+j,\\
	&\mathbf{t}^k=(t^k_0, t^k_1, \dots, t^k_{\ell_1-1}), \,\,\,\,\,\,\,t^k_i=id_2+k,
	\end{aligned}
	\end{equation}
	where $0\leq j,k< m$ and $0\leq i < \ell_1=\ell/m$.
	Denote
	$$\ba = \ba^{0}\Vert\ba^{1}\Vert\cdots \Vert\ba^{2m-1} = \bs^0\Vert\cdots \Vert\bs^{m-1}||\bt^{0}||\cdots||\bt^{m-1}.$$
	Let $\bm{\pi}$ be a permutation of $\{0, 1, \dots, 2m-1\}$.
	Define a sequence $\bu$ of length $n=2\ell$ over $\mathbb{Z}_\ell$ via permuting $\ba$ as follows:
	\begin{equation}\label{key}
	\mathbf{u} = \bm\pi(\ba)=\mathbf{a}^{\pi(0)}\Vert\mathbf{a}^{\pi(1)}\Vert\cdots\Vert\mathbf{a}^{\pi(2m-1)}.
	\end{equation}
\end{cons}

The following example demonstrates Construction \ref{cons_II}.

\begin{example}
	 Take $m=3$.
	For three positive integers $\ell$, $d_1$ and $d_2$ with $\gcd(\ell, d_1) = \gcd(\ell, d_2) = \gcd(\ell, d_2-d_1) = 3$, we obtain $3$ sequences $\bs^j$ of length $\ell/3$ from the matrix $S$ and $3$ sequences $\bt^k$
	from the matrix $T$. Suppose $\bm{\pi} = (0, 1, 3, 2, 4, 5)$.
	Then we have
	$$
	\ba = \bs^{0}\Vert \bs^1 \Vert \bs^2 \Vert \bt^{0}\Vert \bt^1 \Vert \bt^2 \quad  and  \quad \mathbf{u}=\mathbf{s}^{0}\Vert\mathbf{s}^{1}\Vert\mathbf{t}^{0}\Vert\mathbf{s}^{2}\Vert\mathbf{t}^{1}\Vert\mathbf{t}^{2}.
	$$
\end{example}
The Hamming correlation and the minimum gap of the sequence $\bu=\bm{\pi}(\ba)$ in Construction \ref{cons_II} are characterized below.
\begin{theorem}\label{H=H}
Let $\ell$, $d_1$ and $d_2$ be three positive integers with $\gcd(\ell, d_1) = \gcd(\ell, d_2) = \gcd(\ell, d_2-d_1) = m$.
Let $\bm{\pi}$ be a permutation of $\{0,1,2,\dots,2m-1\}$ and $\bm{\pi}_m=(\pi(i)\mod{m})_{0\leq i<2m}$.
Then the sequence $\bu=\bm{\pi}(\ba)$ defined by Construction \ref{cons_II} satisfies
$H_{\bu}=H_{\bm\pi_m}.$ Furthermore,
if $d_1+d_2<\ell-m+2$, then the sequence $\bu$ has minimum gap $d_1-1$.
\end{theorem}

In order to prove Theorem \ref{H=H}, we first discuss the Hamming correlation of sequences $\mathbf{s}^j$ and $\mathbf{t}^k$ in the following lemma.

\begin{lem}\label{H} Let $\mathbf{s}^j$ and $\mathbf{t}^k$, where $0\leq j,k< m$, be defined as in \eqref{st}. Then \vspace{4pt}\\\vspace{3pt}
{\rm(i)} $H_{\mathbf{s}^j}(\tau)=H_{\mathbf{t}^j}(\tau)=0$, $0< \tau< \ell_1;$\\\vspace{3pt}
{\rm(ii)} $H_{\mathbf{s}^j,\mathbf{s}^k}(\tau)=H_{\mathbf{t}^j,\mathbf{t}^k}(\tau)=H_{\mathbf{s}^j,\mathbf{t}^k}(\tau)=0$, $0\leq \tau< \ell_1$ and $j\neq k;$\\\vspace{3pt}
{\rm(iii)} $H_{\mathbf{s}^j,\mathbf{t}^j}(\tau)= 1,$ $0\leq \tau< \ell_1.$
\end{lem}
\begin{proof}
According to the construction of $S$ and $T$, each row in $S$ and $T$ contains distinct elements.
Thus (i) holds and
$H_{\mathbf{s}^j,\mathbf{s}^k}(\tau)=H_{\mathbf{t}^j,\mathbf{t}^k}(\tau)=0$ for $0\leq \tau< l_1$ and $j\neq k$.
In addition, we claim that
$
s^j_{i_1} \neq t^k_{i_2}
$  for any $0\leq i_1, i_2 <l_1$ and $j\neq k$. Otherwise, the equality $s^j_{i_1} = t^k_{i_2}$ implies
$$
i_1d_1 + j \equiv i_2d_2 + k \mod{l},
$$	
indicating that $i_1e_1-i_2e_2$ where $e_i=d_i/m$ for $i = 1, 2$, satisfies the following congruence equation
$$
mx \equiv k-j \mod{l}.
$$
The congruence equation has no solution since $j,k<m$. This is a contradiction.

(iii)
The Hamming cross-correlation of $\mathbf{s}^j$ and $\mathbf{t}^j$ is
$$
H_{\mathbf{s}^j,\mathbf{t}^j}(\tau)
=\sum\limits_{i=0}^{l_1-1}h[id_1+j, (i+\tau)d_2+j]
=\sum\limits_{i=0}^{l_1-1}h[0, (i(e_2-e_1) + \tau e_2) m].
$$
Note that $h[0, (i(e_2-e_1) + \tau e_2) m]=1$ if and only if $l \mid (i(e_2-e_1) + \tau e_2) m$, i.e., $l_1 \mid (i(e_2-e_1) + \tau e_2)$.
Since $\gcd(l, d_2-d_1) = m$, we have $\gcd(l_1, e_1-e_2)=1$.
Hence, for any shift $0\leq \tau<l_1$, there exists exactly one $i \equiv -\tau e_2(e_2-e_1)^{-1}\mod{l_1}$ such that $l \mid (i(e_2-e_1) + \tau e_2)m$, which implies $H_{\mathbf{s}^j,\mathbf{t}^j}(\tau)= 1$ for $0\leq \tau< l_1.$
\end{proof}

With the above preparation, we now give the proof of Theorem \ref{H=H}.

\noindent{\textit{Proof of Theorem \ref{H=H}.}
From Construction \ref{cons_II}, it is readily seen that
$\mathbf{u}$ is a uniform sequence of length $2l$ over $\mathbb{Z}_l$.

We start with the proof of $H_{\mathbf{u}} \leq H_{\bm{\pi}_m}$.
Since $H_{\mathbf{u}}(\tau)=H_{\mathbf{u}}(2l-\tau)$ for $0< \tau<2l$, it suffices to calculate $H_{\mathbf{u}}(\tau)$ for $0<\tau\leq l$.
Depending on the value of $\tau$, we consider two cases: $\tau=rl_1$ and $(r-1)l_1<\tau<rl_1$ for a certain integer $r$ with $0<r\leq m$.

$\mathbf{Case\,\, 1}$. When $\tau=rl_1$, denoting $\bu=\bu^{0}\Vert \bu^1 \Vert \cdots \Vert \bu^{2m-1}$, we have
\begin{align*}
H_{\mathbf{u}}(\tau)
&=\sum\limits_{\kappa=0}^{2m-1}H_{\mathbf{u}^{\kappa},\mathbf{u}^{\kappa+r}}(0),
\end{align*} where the superscripts are taken modulo $2m$. Obviously, $\mathbf{u}^{\kappa}\neq\mathbf{u}^{\kappa+r}$ since $0<r\leq m$.

From Lemma \ref{H} \rm{(iii)} and Construction \ref{cons_II}, $H_{\mathbf{u}^{\kappa},\mathbf{u}^{\kappa+r}}(0)=1$ if and only if
$\{\mathbf{u}^{\kappa},\mathbf{u}^{\kappa+r}\} = \{\bs^j, \bt^j\}$ for some $0\leq j<m$, i.e.,
$\pi_m(\kappa)=\pi_m(\kappa+r)$,
and $H_{\mathbf{u}^{\kappa},\mathbf{u}^{\kappa+r}}(0)=0$ otherwise. Hence
\begin{equation}\label{key}
H_{\mathbf{u}}(\tau) =\sum\limits_{\kappa=0}^{2m-1}H_{\mathbf{u}^{\kappa},\mathbf{u}^{\kappa+r}}(0) = \sum\limits_{\kappa=0}^{2m-1}h[\pi_m(\kappa), \pi_m(\kappa+r)] = H_{\bm{\pi}_m}(r).
\end{equation}

$\mathbf{Case\,\, 2}$. When $(r-1)l_1<\tau<rl_1$, taking $\tau_1=\tau-(r-1)l_1$,  we have $0 < \tau_1 < l_1$ and
\begin{align*}
H_{\mathbf{u}}(\tau)
&=\sum\limits_{\kappa=0}^{2m-1}\left(\sum\limits_{i=0}^{l_1-1-\tau_1}h\left[u^{\kappa}_i,u^{\kappa+r-1}_{i+\tau_1}\right]
+\sum\limits_{i=l_1-\tau_1}^{l_1-1}h\left[u^{\kappa}_i,u^{\kappa+r}_{i+\tau_1-l_1}\right]\right).
\end{align*}
Since $\bm{\pi}_m$ is a uniform FHS of length $2m$ over $\{0,1,\dots, m-1\}$, $\pi_m(\kappa)=\pi_m(\kappa+r-1)$ and $\pi_m(\kappa)=\pi_m(\kappa+r)$ cannot both hold simultaneously when $r \neq 1$.
If $r=1$, then $\pi_m(\kappa)=\pi_m(\kappa+r-1)=\pi_m(\kappa+r)$ is possible, thus
\begin{align*}
H_{\mathbf{u}}(\tau)
&=\sum\limits_{\kappa=0}^{2m-1}\left(\sum\limits_{i=0}^{l_1-1-\tau_1}h\left[u^{\kappa}_i,u^{\kappa}_{i+\tau_1}\right]
+\sum\limits_{i=l_1-\tau_1}^{l_1-1}h\left[u^{\kappa}_i,u^{\kappa+1}_{i+\tau_1-l_1}\right]\right)\\
&=\sum\limits_{\kappa=0}^{2m-1}\sum\limits_{i=l_1-\tau_1}^{l_1-1}h\left[u^{\kappa}_i,u^{\kappa+1}_{i+\tau_1-l_1}\right]\\
&\leq \sum\limits_{\kappa=0}^{2m-1}h[\pi_m(\kappa), \pi_m(\kappa+1)]\\
& = H_{\bm{\pi}_m}(1).
\end{align*}

If $r \neq 1$, let $H_{\bm{\pi}_m}=h$, there are at most $2h$ different $\kappa_v, \kappa'_v \in \mathbb{Z}_{2m}$, $1 \leq v \leq h$, such that
$\pi_m(\kappa_v)=\pi_m(\kappa_v+r-1)$ and $\pi_m(\kappa'_v)=\pi_m(\kappa'_v+r)$. Then
\begin{align*}
H_{\mathbf{u}}(\tau)
&\leq \sum\limits_{v=1}^{h} \left(\sum\limits_{i=0}^{l_1-1-\tau_1}h\left[u^{\kappa_v}_i,u^{\kappa_v+r-1}_{i+\tau_1}\right]
+\sum\limits_{i=l_1-\tau_1}^{l_1-1}h\left[u^{\kappa'_v}_i,u^{\kappa'_v+r}_{i+\tau_1-l_1}\right]\right).
\end{align*}
Based on the proof of Lemma \ref{H} (iii), we know that for $0 \leq j < m$, $h[s^j_i, t^j_{i+\tau}]=1$ if and only if $i \equiv -\tau e_2(e_2-e_1)^{-1}\mod{l_1}$, the one-to-one correspondence between $i$ and $\tau$ is independent of the superscript $j$.
Therefore, for a given $\tau$, $\sum\limits_{i=0}^{l_1-1-\tau_1}h[s^j_i, t^j_{i+\tau}]$ and $\sum\limits_{i=l_1-\tau_1}^{l_1-1}h[s^j_i, t^j_{i+\tau-l_1}]$ cannot both equal $1$ simultaneously.

Similarly,  $h[t^j_i, s^j_{i+\tau}]=1$ if and only if $i \equiv -\tau e_1(e_1-e_2)^{-1}\mod{l_1}$.

Since $H_{\mathbf{u}}=\max\limits_{0< \tau< 2l}H_{\mathbf{u}}(\tau)$, when $(r-1)l_1<\tau<rl_1$, $H_{\mathbf{u}}(\tau)$ can reach its maximum value if all $h[s^j_i, t^j_{i+\tau}]$ and $h[t^j_i, s^j_{i+\tau}]$ are distributed as much as possible across two different summation ranges.
This means that
\begin{align*}
H_{\mathbf{u}}(\tau)
&\leq \sum\limits_{v=1}^{h} \left( \sum\limits_{i=0}^{l_1-1-\tau_1}h\left[s^{\pi_m(\kappa_v)}_i,t^{\pi_m(\kappa_v+r-1)}_{i+\tau_1}\right]
+\sum\limits_{i=l_1-\tau_1}^{l_1-1}h\left[t^{\pi_m(\kappa'_v)}_i,s^{\pi_m(\kappa'_v+r)}_{i+\tau_1-l_1}\right]\right)
\end{align*}
or
\begin{align*}
H_{\mathbf{u}}(\tau)
&\leq \sum\limits_{v=1}^{h} \left( \sum\limits_{i=0}^{l_1-1-\tau_1}h\left[t^{\pi_m(\kappa_v)}_i,s^{\pi_m(\kappa_v+r-1)}_{i+\tau_1}\right]
+\sum\limits_{i=l_1-\tau_1}^{l_1-1}h\left[s^{\pi_m(\kappa'_v)}_i,t^{\pi_m(\kappa'_v+r)}_{i+\tau_1-l_1}\right]\right).
\end{align*}
The proofs of the two inequalities above are similar, we will only present the proof of the first one.

Suppose there are $i_0 \in [0, l_1-1-\tau_1]$ and $i_1 \in [l_1-\tau_1, l_1-1]$ such that
$h\left[s^{j}_{i_0},t^{j}_{i_0+\tau_1}\right]=1$ and $h\left[t^{k}_{i_1},s^{k}_{i_1+\tau_1-l_1}\right]=1$, equivalently,
$$
i_0 \equiv -\tau_1 e_2(e_2-e_1)^{-1}\mod{l_1}
$$
and
$$
i_1 \equiv -(\tau_1-l_1) e_1(e_1-e_2)^{-1}\mod{l_1}.
$$
Then $i_0+i_1 \equiv -\tau_1 \mod{l_1}$. Since $i_0+i_1 \in [l_1-\tau_1, 2l_1-\tau_1-2]$, we have $i_0+i_1 = l_1 -\tau_1$, which implies $i_0=0$ and $i_1 = l_1-\tau_1$.
Thus, $l_1 \mid -\tau_1 e_2(e_2-e_1)^{-1}$. Since $\gcd(l_1, e_2)=\gcd(l_1, e_2-e_1)=1$, we have $l_1 \mid \tau_1$.
This contradicts $0<\tau_1<l_1$.
Hence $\sum\limits_{i=0}^{l_1-1-\tau_1}h\left[s^{\pi_m(\kappa_v)}_{i},t^{\pi_m(\kappa_v+r-1)}_{i+\tau_1}\right]=1$ and $\sum\limits_{i=l_1-\tau_1}^{l_1-1}h\left[t^{\pi_m(\kappa'_v)}_{i},s^{\pi_m(\kappa'_v+r)}_{i+\tau_1-l_1}\right]=1$ cannot both hold simultaneously.
Therefore, $H_{\mathbf{u}}(\tau)\leq h = H_{\bm{\pi}_m}$.

Now we will show $ H_{\bm{\pi}_m} \leq H_{\mathbf{u}}$.
According to Lemma \ref{H}, $H_{\mathbf{s}^j, \mathbf{t}^j}=1$ for $0\leq j < m$, otherwise, the Hamming correlation of those length-$l_1$ subsequences is $0$.
 Assume $H_{\mathbf{u}}=h'$.
 For any $0< \tau <2l$, there are at most $h'$ different superscripts $0\leq j_z< m$ and  $1\leq z \leq h'$ such that $H_{\mathbf{s}^{j_z}, \mathbf{t}^{j_z}}(\tau)=1$. That is to say, for any time delay $0< \tau <2l$, at most $h'$ pairs of subsequences with the same superscripts can overlap in the sequence $\mathbf{u}$.
Additionally, the superscripts of these length-$l_1$ subsequences in the sequence $\mathbf{u}$ correspond to the sequence $\bm{\pi}_m$. It follows that $H_{\bm{\pi}_m}(\tau)\leq h'$ for $0< \tau <2l$.
Note that $m < l$, then $H_{\bm{\pi}_m}(\tau)\leq h'$ for $0< \tau <2m$, i.e., $ H_{\bm{\pi}_m} \leq H_{\mathbf{u}}$.

Combing the above analysis we obtain
$$
 H_{\bm{\pi}_m} = H_{\mathbf{u}}.
$$

For the gaps of sequences $\mathbf{s}^j$ and $\mathbf{t}^k$, it is clear that
$$|s^j_{i+1}-s^j_i|>d_1-1\,\,\, \text{and} \,\,\,|t^k_{i+1}-t^k_i|>d_2-1,$$
where $0\leq j,\,\,k< m$ and $0\leq i\leq l_1-2$.

From Construction \ref{cons_II}, the possible concatenations in $\mathbf{u}$ are $\mathbf{s}^j\Vert\mathbf{s}^k$, $\mathbf{s}^j\Vert\mathbf{t}^k$, $\mathbf{t}^j\Vert\mathbf{s}^k$ and $\mathbf{t}^j\Vert\mathbf{t}^k$, where $0\leq j,\,\,k< m$.
For the gaps at the concatenating positions, we have
\begin{equation*}\label{key}
\begin{array}{l}
\left|s^j_{l_1-1}-s^k_0\right|=\left|(l_1-1)d_1+j-k\right|=|l-d_1+j-k|, \\
\left|s^j_{l_1-1}-t^k_0\right|=\left|(l_1-1)d_1+j-k\right|=|l-d_1+j-k|, \\
\left|t^j_{l_1-1}-s^k_0\right|=\left|(l_1-1)d_2+j-k\right|=|l-d_2+j-k|, \\
\left|t^j_{l_1-1}-t^k_0\right|=\left|(l_1-1)d_2+j-k\right|=|l-d_2+j-k|,
\end{array}
\end{equation*}
where $j\neq k$ in the first and last equations, $(l_1-1)d_1 = le_1 - d_1 \equiv l -d_1 \mod l$ and $(l_1-1)d_2 = le_2 - d_2 \equiv l -d_2 \mod l$.

Next, we will discuss the minimum value of the four aforementioned gaps. For given positive integers $l$, $d_1$, $d_2$ with $d_1<d_2$ and $\gcd(l, d_1) = \gcd(l, d_2) = \gcd(l, d_2-d_1) = m$, we have $l - d_1 > l - d_2 \geq m$. Note that $-(m-1)\leq j-k\leq m-1$. Then $(l-d_1+j-k)>(l-d_2+j-k)>0$. It follows that
$$
\min_{0\leq j,k<m}\{|l-d_1+j-k|, |l-d_2+j-k|\}=\min_{0\leq j,k<m}\{l-d_2+j-k\}.
$$
In order to find the minimum gap, it is sufficient to minimize $j-k$, i.e., $j=0$ and $k=m-1$. Thus,
$$
\min_{0\leq j,k<m}\{l-d_2+j-k\}=l-d_2-m+1.
$$
As a result, $l-d_2-m+1>d_1-1$ since $d_1+d_2<l-m+2$. Hence the sequence $\mathbf{u}$ has minimum gap $d_1-1$. \hfill$\square$

From Theorem \ref{H=H}, we know that the Hamming correlation of sequences in Construction \ref{cons_II} is equal to the Hamming correlation of $\bm{\pi}_m$. Based on this, the optimal uniform WGFHSs with parameters $(2l, l, 2)$ can be determined.
\begin{coro}\label{H=H=2}
Let $\ell$, $d_1$, $d_2$ and $m$ be four positive integers with $\gcd(\ell, d_1) = \gcd(\ell, d_2) = \gcd(\ell, d_2-d_1) = m$. Then the sequence $\mathbf{u}$ defined by Construction \ref{cons_II} is a
$(2\ell, \ell, 2)$-FHS with minimum gap $d_1-1$, where $d_1+d_2<\ell-m+2$,
 if and only if $\bm{\pi}_m$ is a uniform $(2m, m, 2)$-FHS.
\end{coro}

\begin{remark}
	Notice that for any permutation $\bm{\pi}$ of $\{0,1,\dots, 2m-1\}$, we obtain a unique length-$2m$ sequence $\bm{\pi}_m$.
	On the other hand, given a uniform  length-$2m$  sequence over $\{0,1,\dots, m-1\}$, denoted as $\bm{\pi}_m$,  there are in total $2^m$
	different permutations $\bm{\pi}$ corresponding to $\bm{\pi}_m$.
	Therefore, one optimal uniform $(2m, m, 2)$-FHS $\bm\pi_m$ will yield $2^m$ optimal $(2\ell, \ell, 2)$-FHSs with minimum gap $d_1-1$
	when $d_1$ and $d_2$ are properly chosen.
\end{remark}

The following example demonstrates Theorem \ref{H=H} and Corollary \ref{H=H=2}.

\begin{example}\label{eg}

For $m=3$, Table \ref{Tab m=3} lists all $\bm{\pi}_m$ that can be used to construct optimal FHSs with length $2\ell$, where $m \mid \ell$.
\begin{table}[H]\vspace {-0.5em}
\label{Tab m=3}\centering
\caption{Optimal\,\,\,uniform\,\,\,$\bm{\pi}_3$}
\label{Tab m=3}
\begin{tabular}{|c|}
  \hline
   $\bm{\pi}_3$ \\
  \hline
  $(0, 0, 1, 2, 1, 2)$ \,\,\,$(0, 0, 1, 2, 2, 1)$ \,\,\,$(0, 0, 2, 1, 1, 2)$ \,\,\,$(0, 0, 2, 1, 2, 1)$\\
  $(0, 1, 0, 1, 2, 2)$ \,\,\,$(0, 1, 0, 2, 1, 2)$ \,\,\,$(0, 1, 0, 2, 2, 1)$ \,\,\,$(0, 1, 1, 0, 2, 2)$\\
  $(0, 1, 1, 2, 0, 2)$ \,\,\,$(0, 1, 2, 0, 2, 1)$ \,\,\,$(0, 1, 2, 1, 0, 2)$ \,\,\,$(0, 2, 0, 2, 1, 1)$\\
  \hline
\end{tabular}
\end{table}\vspace {-1.0em}
Each $\bm{\pi}_3$ corresponds to 8 different permutations $\bm{\pi}$ on $\{0,1,\dots, 5\}$.
These permutations are used to choose $\bs^j$ and $\bt^k$ from the matrices $S$ and $T$, respectively, to obtain the sequence $\bu$.
To demonstrate the relationship more clearly, we will show the specific result of $\bm{\pi}_3 = (0, 0, 1, 2, 1, 2)$.

For $\bm{\pi}_3 = (0, 0, 1, 2, 1, 2)$, there are 8 different permutations $\bm{\pi}$:
 $$(0, 3, 1, 2, 4, 5), \,\,\,\,(3, 0, 1, 2, 4, 5), \,\,\,\,(0, 3, 4, 2, 1, 5), \,\,\,\,(0, 3, 1, 5, 4, 2),$$
 $$(3, 0, 4, 2, 1, 5), \,\,\,\,(3, 0, 1, 5, 4, 2), \,\,\,\,(0, 3, 4, 5, 1, 2), \,\,\,\,(3, 0, 4, 5, 1, 2).$$
These permutations correspond to the following sequences $\bu$:
 $$\mathbf{s}^0\Vert\mathbf{t}^0\Vert\mathbf{s}^1\Vert\mathbf{s}^2\Vert\mathbf{t}^1
\Vert\mathbf{t}^2, \quad \mathbf{t}^0\Vert\mathbf{s}^0\Vert\mathbf{s}^1\Vert\mathbf{s}^2\Vert\mathbf{t}^1
\Vert\mathbf{t}^2, \quad \mathbf{s}^0\Vert\mathbf{t}^0\Vert\mathbf{t}^1\Vert\mathbf{s}^2\Vert\mathbf{s}^1
\Vert\mathbf{t}^2, \quad \mathbf{s}^0\Vert\mathbf{t}^0\Vert\mathbf{s}^1\Vert\mathbf{t}^2\Vert\mathbf{t}^1
\Vert\mathbf{s}^2,$$
 $$\mathbf{t}^0\Vert\mathbf{s}^0\Vert\mathbf{t}^1\Vert\mathbf{s}^2\Vert\mathbf{s}^1
\Vert\mathbf{t}^2, \quad \mathbf{t}^0\Vert\mathbf{s}^0\Vert\mathbf{s}^1\Vert\mathbf{t}^2\Vert\mathbf{t}^1
\Vert\mathbf{s}^2, \quad \mathbf{s}^0\Vert\mathbf{t}^0\Vert\mathbf{t}^1\Vert\mathbf{t}^2\Vert\mathbf{s}^1
\Vert\mathbf{s}^2, \quad \mathbf{t}^0\Vert\mathbf{s}^0\Vert\mathbf{t}^1\Vert\mathbf{t}^2\Vert\mathbf{s}^1
\Vert\mathbf{s}^2.$$
Without loss of generality, we consider $\bm{\pi} = (0, 3, 1, 2, 4, 5)$ and $\bu = \mathbf{s}^0\Vert\mathbf{t}^0\Vert\mathbf{s}^1\Vert\mathbf{s}^2\Vert\mathbf{t}^1
\Vert\mathbf{t}^2$.

{\rm(i)}  When $\ell=21$, $d_1=6$ and $d_2=9$, from Construction \ref{cons_II}, we have
$$\mathbf{s}^0=(0, 6, 12, 18, 3, 9, 15),\,\,\, \mathbf{s}^1=(1, 7, 13, 19, 4, 10, 16),\,\,\, \mathbf{s}^2=(2, 8, 14, 20, 5, 11, 17),$$
$$\mathbf{t}^0=(0, 9, 18, 6, 15, 3, 12),\,\,\, \mathbf{t}^1=(1, 10, 19, 7, 16, 4, 13),\,\,\, \mathbf{t}^2=(2, 11, 20, 8, 17, 5, 14)$$
and
 \[
\begin{split}
\mathbf{u} &= (0, 6, 12, 18, 3, 9, 15, 0, 9, 18, 6, 15, 3, 12, 1, 7, 13, 19, 4, 10, 16, 2, 8, 14, 20, 5, 11, 17, 1, 10, \\
&\,\,\,\,\,\,\,\,\,19, 7, 16, 4, 13, 2, 11, 20, 8, 17, 5, 14).
\end{split}
\]
Obviously, the sequence $\mathbf{u}$ is an FHS with length $42$ and minimum gap $5$ over $\mathbb{Z}_{21}$.
We can calculate that $H_{\mathbf{u}}(\tau)\leq 2$ for $0< \tau <42$. Thus, $\mathbf{u}$ is an optimal $(42, 21, 2)$-WGFHS with minimum gap $5$.

{\rm(ii)}  When $\ell=15$, $d_1=6$ and $d_2=9$, we note that $d_1+d_2 > l-m+2$ and the subsequences are as follows:
$$\mathbf{s}^0=(0, 6, 12, 3, 9),\,\,\, \mathbf{s}^1=(1, 7, 13, 4, 10),\,\,\, \mathbf{s}^2=(2, 8, 14, 5, 11),$$
$$\mathbf{t}^0=(0, 9, 3, 12, 6),\,\,\, \mathbf{t}^1=(1, 10, 4, 13, 7),\,\,\, \mathbf{t}^2=(2, 11, 5, 14, 8).$$
For the FHS
\[
\begin{split}
\mathbf{u} = (0, 6, 12, 3, 9, 0, 9, 3, 12, 6, 1, 7, 13, 4, 10, 2, 8, 14, 5, 11, 1, 10, 4, 13, 7, 2, 11, 5, 14, 8),
\end{split}
\]
it is easy to verify that $\mathbf{u}$ is optimal with respect to the Lempel-Greenberger bound, and that its minimum gap is $4$, rather than $d_1-1=5$.

The above discussion is consistent with Theorem \ref{H=H} and Corollary \ref{H=H=2}.

\end{example}

\begin{remark}
 In Theorem \ref{H=H}, when $d_1 = \lfloor \frac{\ell-1}{2} \rfloor$ and $d_2 = d_1+1$, the minimum gap can reach the tight upper bound $\lfloor \frac{\ell-1}{2} \rfloor-1$. This result is consistent with reference \cite{PC22}.
Since $\gcd(\ell, d_1) = \gcd(\ell, d_2) = \gcd(\ell, d_2-d_1) = m$ and $d_1+d_2< \ell-m+2$, we can deduce that $d_1< \frac{\ell}{2}-m+1$. This implies that the minimum gap is less than $\frac{\ell}{2}-m$.
Consequently, as the common factor $m$ increases, the minimum gap of FHSs in Theorem \ref{H=H} decreases.
\end{remark}

It is worth noting that starting from a different matrix $A' = A+k$, the $j$-th row of $A'$ is shift-equivalent to the $(j+k)_m$-th row of $A$, where $0\leq k<l$, $0\leq j \leq m-1$, and $A+k$ represents the operation of adding $k$ to every entry in the matrix $A$.
Let $S'=\bm{\rho}_1(A') = [\mathbf{s}'^0, \mathbf{s}'^1, \dots, \mathbf{s}'^{m-1}]^T$ and $T'=\bm{\rho}_2(A') = [\mathbf{t}'^0, \mathbf{t}'^1, \dots, \mathbf{t}'^{m-1}]^T$. Based on this, Construction \ref{cons_II} will yield the FHS $\bu'$ with a permutation $\bm{\pi}'$ on $\{0, 1, \dots, 2m-1\}$. Because the $j$-th row of $S'$, $T'$ is shift-equivalent to the $(j+k)_m$-th row of $S$, $T$, respectively, the sequences $\mathbf{s}'^j$ and $\mathbf{t}'^j$ also satisfy Lemma \ref{H}. Then $H_{\mathbf{u}'} = H_{\bm{\pi}'_m}$.

\subsection{Applications from known constructions}
Given an optimal uniform $(2m, m, 2)$-FHS $\bm{\pi_m}$ and selecting proper parameters $d_1$ and $d_2$, the optimal uniform $(2l, l, 2)$-WGFHS $\bu$ can be constructed by Construction \ref{cons_II}.
The following discusses how this construction is explicitly applied to existing optimal FHSs of even length.

\textit{Application I (Construction $B1$ from \cite{JK09})}:
let $\phi$ be a permutation on $\mathbb{Z}_N$ and $DU({\mathbb{Z}_N})$ a difference unit
set of $\mathbb{Z}_N$, where $N=p_1^{e_1}p_2^{e_2}\cdots p_r^{e_r}$. For $i\in \{0\}\cup DU({\mathbb{Z}_N})$, define a sequence $\mathbf{x}_i=\{x_i(t)\}^{N-1}_{t=0}$
 as follows:
\begin{align*}
x_i(t)=
\begin{cases}
0, \,\,\,\,\,\, \,\,\,\,\,\, if \,\,\,i = 0,\\
\phi(it), \,\,\,if \,\,\,i\in DU(\mathbb{Z}_N).
\end{cases}
\end{align*}
Let $\varepsilon$ be a one-to-one map from $\mathbb{Z}_k$ to $DU(\mathbb{Z}_N)$ for $k = p_1-1$. Let $\mathbf{y}=\{y(t)\}^{kN-1}_{t=0}$
be the FHS of length $kN$ over $\mathbb{Z}_N$ defined by
$$
y(t) = y(t_0,t_1) = x_{\varepsilon(t_0)}(t_1+\gamma(t_0)),
$$
where $t_0=(t)_k$, $t_1=(t)_N$, and $\gamma$ is an arbitrary function from $\mathbb{Z}_k$
to $\mathbb{Z}_N$. Then the sequence $\mathbf{y}$ is an optimal uniform $(kN, N, k)$-FHS.

Let $k=2$, $\phi(i)=i$, $\varepsilon(0)=1$, $\varepsilon(1)=2$ and $\gamma(t_0)=0$.
Then we obtain a length-$2N$ FHS $\mathbf{y} = (x_1(0), x_2(1), x_1(2), x_2(3), \dots, x_1(N-1), x_2(0), x_1(1), x_2(2), x_1(3), \dots, x_2(N-1)) = (0, 2, 2, 6, \dots, N-1, 0, 1, 4, 3, \dots, 2(N-1))$ over $\mathbb{Z}_N$.

For $N=9$, we have $\mathbf{y}=(0, 2, 2, 6, 4, 1, 6, 5, 8, 0, 1, 4, 3, 8, 5, 3, 7, 7)$.
Let $\bm{\pi}_N = \mathbf{y}$. Note that there are $2^N$ different permutations $\bm{\pi}$ corresponding to $\bm{\pi}_N$.
 These permutations can be discussed similarly.
Without loss of generality, let $\bm{\pi}=(0, 2, 11, 6, 4, 1, 15, 5, 8, 9, 10,\\
 13, 3, 17, 14, 12, 7, 16)$.

When $l=27$, $d_1=9$ and $d_2=18$,  from Construction \ref{cons_II}, we have $\mathbf{s}^j=(j, j+9, j+18)$, $\mathbf{t}^k=(k, k+18, k+9)$, $0 \leq j, k \leq 8$ and
\begin{align*}
\bu
&=\mathbf{s}^0\Vert\mathbf{s}^2\Vert\mathbf{t}^2\Vert\mathbf{s}^6\Vert\mathbf{s}^4\Vert\mathbf{s}^1\Vert
\mathbf{t}^6\Vert\mathbf{s}^5\Vert\mathbf{s}^8\Vert\mathbf{t}^0\Vert\mathbf{t}^1\Vert\mathbf{t}^4\Vert\mathbf{s}^3
\Vert\mathbf{t}^8\Vert\mathbf{t}^5\Vert\mathbf{t}^3\Vert\mathbf{s}^7\Vert\mathbf{t}^7\\
&=(0, 9, 18, 2, 11, 20, 2, 20, 11, 6, 15, 24, 4, 13, 22, 1, 10, 19, 6, 24, 15, 5, 14, 23, 8, 17, 26,\\
&\,\,\,\,\,\,\,\,\,0, 18, 9, 1, 19, 10, 4, 22, 13, 3, 12, 21, 8, 26, 17, 5, 23, 14, 3, 21, 12, 7, 16, 25, 7, 25, 16).
\end{align*}
The sequence $\bu$ is an optimal uniform $(54, 27, 2)$-FHS with minimum gap $4$.

\textit{Application II (Theorem 4 from \cite{CJ08})}:
let $q = ef + 1$ be a prime power and $C_0$, $C_1$, $\dots$, $C_{e-1}$ be the cyclotomic classes of order $e$ with respect to $\mathbb{F}_q$ defined by a generator $\omega$ of $\mathbb{F}^*_q$. Define the function from $\mathbb{F}_q$ to $\mathbb{Z}_{ef}$ by
\begin{align*}
\xi(x)=
\begin{cases}
l, \,\,\, \,\,\, \,\,\,\,\,\,\,\,\,\,\, \,\,\, \,\,\,\,\,\,\,\,\,\,if \,\,\,x = 1,\\
\log_\omega(x-1), \,\,\,if \,\,\,x\neq 1,
\end{cases}
\end{align*}
where $l=(q-1)/2$ if $q$ is odd, and $l=0$ if $q$ is even.
For $0\leq z \leq e-1$, define $D_z = \xi(C_z) = \{\xi(x): x\in C_z\}\subset\mathbb{Z}_{ef}$.

Construct an FHS $\mathbf{y}=(y(0), y(1), \dots, y(q-2))$ of length $q-1$ over the alphabet $\mathbb{Z}_e$,
where
$y(i)=z\in \mathbb{Z}_e$ if and only if $i\in D_{z}$ for $0\leq i \leq q-2$.
Then $\mathbf{y}$ is an optimal uniform $(ef, e, f)$-FHS.

For example, setting $q=5^2$, $e=12$ and $\omega$ be the primitive element of $\mathbb{F}_{5^2}$ with minimal polynomial $x^2+4x+2$ over $\mathbb{F}_5$, we obtain an optimal uniform FHS: $\mathbf{y}=(6, 10, 5, 10, 11, 2, 6, 8, 4, 11, 1, 4, 0, 5, 3, 2, 8, 1, 0, 9, 7, 7, 3, 9)$.
Similarly as Application I, we choose $\bm{\pi}=(6, 10, 5, 22, 11, 2, 18, 8, 4, 23, 1, 16, 0, 17, 3, 14, 20, 13, 12, 9, 7, 19, 15, 21)$.

From Construction \ref{cons_II}, letting $l=36$, $d_1=12$ and $d_2=24$, we have $\mathbf{s}^j=(j, j+12, j+24)$, $\mathbf{t}^k=(k, k+24, k+12)$, $0 \leq j, k \leq 11$ and
\begin{align*}
\bu
&=\mathbf{s}^6\Vert\mathbf{s}^{10}\Vert\mathbf{s}^5\Vert\mathbf{t}^{10}\Vert\mathbf{s}^{11}\Vert\mathbf{s}^2\Vert
\mathbf{t}^6\Vert\mathbf{s}^8\Vert\mathbf{s}^4\Vert\mathbf{t}^{11}\Vert\mathbf{s}^1\Vert\mathbf{t}^4\Vert
\mathbf{s}^0\Vert\mathbf{t}^5\Vert\mathbf{s}^3\Vert\mathbf{t}^2\Vert\mathbf{t}^8\Vert\mathbf{t}^1\Vert
\mathbf{t}^0\Vert\mathbf{s}^9\Vert\mathbf{s}^7\Vert\mathbf{t}^7\Vert\mathbf{t}^3\Vert\mathbf{t}^9\\
&=(6, 18, 30, 10, 22, 34, 5, 17, 29, 10, 34, 22, 11, 23, 35, 2, 14, 26, 6, 30, 18, 8, 20, 32, 4, 16, 28,\\
 &\,\,\,\,\,\,\,\,\,11, 35, 23, 1, 13, 25, 4, 28, 16, 0, 12, 24, 5, 29, 17, 3, 15, 27, 2, 26, 14, 8, 32, 20, 1, 25, 13, 0,\\
  &\,\,\,\,\,\,\,\,\, 24, 12, 9, 21, 33, 7, 19, 31, 7, 31, 19, 3, 27, 15, 9, 33, 21).
\end{align*}
The sequence $\bu$ is an optimal uniform $(72, 36, 2)$-FHS with minimum gap $2$.

\textit{Application III (Construction $\mathcal{B}$ from \cite{JK10})}:
let $p$ be an odd prime and $\alpha$ a primitive element of $\mathbb{F}_p$, then $\mathbb{F}_p$ can be decomposed into $\mathbb{F}_p = \{0\} \cup D_0 \cup D_1$ where
\begin{align*}
&D_0=\left\{\alpha^{2i} | 0\leq i \leq \frac{p-1}{2}-1\right\},\\
&D_1=\left\{\alpha^{2i+1} | 0\leq i \leq \frac{p-1}{2}-1\right\}.
\end{align*}
For any positive integer $k$ with $2\leq k < p$, let $B=\{B(t)\}^{k-1}_{t=0}$ be a binary sequence
consisting of $0$ and $1$ satisfying $\sum\limits_{t=0}^{k-1}(-1)^{B(t)-B(t+\tau)} \leq 0$ for any $1 \leq \tau \leq k-1$.
Let $\mathbf{x}=\{x(t)\}^{k-1}_{t=0}$ be a sequence over $\mathbb{F}_p$ satisfying
\begin{align*}
x(t)\in
\begin{cases}
D_0,\,\,\, if \,\,\, B(t)=0,\\
D_1,\,\,\, if \,\,\, B(t)=1,
\end{cases}
\end{align*}
and $x(t_i) \neq x(t_j)$ for any $0 \leq t_i \neq t_j \leq k-1$.
The FHS $\mathbf{y}=\{y(t)\}^{kp-1}_{t=0}$ over $\mathbb{F}_p$ is defined as
$$
y(t) = y(t_0,t_1) = x(t_0)\cdot t^2_1,
$$
where $t_0 = (t)_k$ and $t_1 = (t)_p$.
It is worth mentioning that the sequences $\mathbf{y}$ in Construction $\mathcal{B}$ are not always uniform, but a uniform sequence can be obtained with a proper selection of parameters.

For example, let $k=2$, $p=5$, $B=(0, 1)$ and $\mathbb{F}_5 = \{0\} \cup D_0 \cup D_1$, where $D_0=\{1, 4\}$ and $D_1=\{2, 3\}$. Then selecting $\mathbf{x}=(1, 3)$ gives $\mathbf{y}=(0, 3, 4, 2, 1, 0, 1, 2, 4, 3)$, which is an optimal uniform $(10, 5, 2)$-FHS.
Based on this, we can construct optimal uniform $(2l,l,2)$-WGFHSs by Construction \ref{cons_II}.

For $l=25$, $d_1=5$ and $d_2=15$, we have $m=5$, $d_1+d_2 < l-m+2$ and $\bm{\pi}_5=(0, 3, 4, 2, 1, 0, 1, 2, 4, 3)$.
Without loss of generality, we take $\bm{\pi}=(0, 3, 4, 2, 1, 5, 6, 7, 9, 8)$, and then
\begin{align*}
\bu&=\mathbf{s}^0\Vert\mathbf{s}^3\Vert\mathbf{s}^4\Vert\mathbf{s}^2\Vert\mathbf{s}^1\Vert\mathbf{t}^0\Vert
\mathbf{t}^1\Vert\mathbf{t}^2\Vert\mathbf{t}^4\Vert\mathbf{t}^3\\
&=(0, 5, 10, 15, 20, 3, 8, 13, 18, 23, 4, 9, 14, 19, 24, 2, 7, 12, 17, 22, 1, 6, 11, 16, 21, 0, 15, 5, 20, \\
&\,\,\,\,\,\,\,\,\,10, 1, 16, 6, 21, 11, 2, 17, 7, 22, 12, 4, 19, 9, 24, 14, 3, 18, 8, 23, 13).
\end{align*}
It can be easily verified that $\bu$ is an optimal uniform $(50, 25, 2)$-FHS with minimum gap $d_1-1 = 4$.

At the end of this section, we make a comparison between the FHSs constructed in \cite{JK09,CJ08,JK10} and the FHSs in this paper. After selecting proper parameters and functions, the first three FHSs can be employed as concatenation-ordering sequences to generate the one in Construction \ref{cons_II}.
In the first three FHSs, the minimum gap was not taken into consideration, and it can be minus one.  More concretely,
 in Application I, for some $0\leq t \leq kN-1$, if $\varepsilon(t_0)(t_1+\gamma(t_0)) \equiv \varepsilon(t_0+1)(t_1+1+\gamma(t_0+1)) \mod N$, then $y_t = y_{t+1}$, where the operation $t_0+1$ is taken modulo $k$.
In Application II, for $0\leq z\leq e-1$, $0\leq j_1\neq j_2 \leq f-1$ and $(z, j_1), (z, j_2) \neq (0,0)$, when the generator $\omega$ of $\mathbb{F}^*_q$ satisfies $\omega^{z+j_2e+1} - \omega = \omega^{z+j_1e} - 1$,  the sequence $\mathbf{y}$ will have two adjacent frequencies that are the same.
This can be resolved by applying Construction \ref{cons_II}, from which one can obtain sequences with minimum gap $d_1-1$ when $d_1$ and $d_2$ are properly chosen. The following table summarises the comparison among sequences from these constructions.
\begin{table}[H]
\label{Tab FHSs}\centering
\caption{Comparison of optimal $(n, l, \lambda)$-FHSs}
\small
\label{Tab FHSs}
\begin{tabular}{cccc}
  \hline
   Parameters & Constraints & Minimum gaps & References \\
  \hline
  $(ef, e, f)$ & $q = ef+1$ is a prime power & uncontrolled & \cite{CJ08}\\

  $(kN, N, k)$ & $\setminus$ & uncontrolled & \cite{JK09}\\

  $(kp, p, k)$ & $p$ is an odd prime & uncontrolled & \cite{JK10}\\

  \multirow{2}*{$(2l, l, 2)$} & $\gcd(l, d_1) = \gcd(l, d_2) = \gcd(l, d_2-d_1) = m$, & \multirow{2}*{$d_1-1$} & \multirow{2}*{Theorem \ref{H=H}}\\
      &  $d_1+d_2<l-m+2$  &   &   \\

  \hline
\end{tabular}
\end{table}

\section{Conclusion}
The main objective of this paper is to construct optimal frequency-hopping sequences with controlled minimum gaps. We began with providing tight upper bounds on the minimum gaps of uniform sequences.
We presented a general construction of optimal sequences of length $2l$ and $3l$ for $d_1, d_2$ and $d_3$ taken from a difference unit set, which contains the results of \cite{PC22} as a special case.
Furthermore, we extended it to the case where $\gcd(l, d_1) = \gcd(l, d_2) = \gcd(l, d_2-d_1) = m \geq 2$.
We proposed a recursive construction producing length-$2l$ optimal FHSs from length-$2m$ optimal FHSs, where $m$ divides $l$. This approach allows us to obtain optimal FHSs with controlled minimum gaps from existing constructions of optimal FHSs in \cite{JK09, CJ08, JK10}.

\section*{Acknowledgment}

The authors sincerely thank the anonymous reviewers for their constructive comments and suggestions that significantly improve the quality of this paper. The authors also thank the associate editor, Dr. Gohar Kyureghyan, for the efficient management of the review process.


\begin{thebibliography}{1}
	
\bibitem{MJ94}
M.~K. Simon, J.~K. Omura, R.~A. Scholtz, and B.~K. Levitt, \emph{Spread
  spectrum communications handbook}.\hskip 1em plus 0.5em minus 0.4em\relax New
  York: McGraw-Hill, 1994.

\bibitem{AH74}
A.~Lempel and H.~Greenberger, ``Families of sequences with optimal
  {Hamming}-correlation properties,'' \emph{IEEE Trans. Inf. Theory}, vol.~20,
  no.~1, pp. 90--94, Jan. 1974.

\bibitem{DP04}
D.~Peng and P.~Fan, ``Lower bounds on the {Hamming} auto- and cross
  correlations of frequency-hopping sequences,'' \emph{IEEE Trans. Inf.
  Theory}, vol.~50, no.~9, pp. 2149--2154, Sep. 2004.

\bibitem{RY04}
R.~Fuji-Hara, Y.~Miao, and M.~Mishima, ``Optimal frequency hopping sequences: a
  combinatorial approach,'' \emph{IEEE Trans. Inf. Theory}, vol.~50, no.~10,
  pp. 2408--2420, Oct. 2004.

\bibitem{WC05}
W.~Chu and C.~J. Colbourn, ``Optimal frequency-hopping sequences via cyclotomy,''
  \emph{IEEE Trans. Inf. Theory}, vol.~51, no.~3, pp. 1139--1141, Mar. 2005.

\bibitem{CM07}
C.~Ding, M.~J. Moisio, and J.~Yuan, ``Algebraic constructions of optimal
  frequency-hopping sequences,'' \emph{IEEE Trans. Inf. Theory}, vol.~53,
  no.~7, pp. 2606--2610, Jul. 2007.

\bibitem{GY09}
G.~Ge, Y.~Miao, and Z.~Yao, ``Optimal frequency hopping sequences: Auto- and
  cross-correlation properties,'' \emph{IEEE Trans. Inf. Theory}, vol.~55,
  no.~2, pp. 867--879, Feb. 2009.

\bibitem{JK09}
J.-H. Chung, Y.~K. Han, and K.~Yang, ``New classes of optimal frequency-hopping
  sequences by interleaving techniques,'' \emph{IEEE Trans. Inf. Theory},
  vol.~55, no.~12, pp. 5783--5791, Dec. 2009.

\bibitem{JK10}
J.-H. Chung and K.~Yang, ``Optimal frequency-hopping sequences with new
  parameters,'' \emph{IEEE Trans. Inf. Theory}, vol.~56, no.~4, pp. 1685--1693,
  Apr. 2010.

\bibitem{XH12}
X.~Zeng, H.~Cai, X.~Tang, and Y.~Yang, ``A class of optimal frequency hopping
  sequences with new parameters,'' \emph{IEEE Trans. Inf. Theory}, vol.~58,
  no.~7, pp. 4899--4907, Jul. 2012.

\bibitem{XH13}
X.~\vspace{0mm}Zeng, H.~Cai, X.~Tang, and Y.~Yang, ``Optimal frequency hopping
  sequences of odd length,'' \emph{IEEE Trans. Inf. Theory}, vol.~59, no.~5,
  pp. 3237--3248, May 2013.

\bibitem{CJ08}
C.~Ding and J.~Yin, ``Sets of optimal frequency-hopping sequences,'' \emph{IEEE
  Trans. Inf. Theory}, vol.~54, no.~8, pp. 3741--3745, Aug. 2008.

\bibitem{CR09}
C.~Ding, R.~Fuji-Hara, Y.~Fujiwara, M.~Jimbo, and M.~Mishima, ``Sets of
  frequency hopping sequences: Bounds and optimal constructions,'' \emph{IEEE
  Trans. Inf. Theory}, vol.~55, no.~7, pp. 3297--3304, Jul. 2009.

\bibitem{WZ02}
W.~Mei and Z.~Zhang, ``Construction of families of {FH} sequences with given
  minimum gap,'' \emph{Chinese Journal of Radio Science}, vol.~17, no.~1, pp.
  16--22, Jan. 2002.

\bibitem{Chen83}
W.~Chen, ``Frequency hopping patterns with wide intervals,'' \emph{Journal of
  System Science and Mathematical Science Chinese Series}, vol.~3, no.~4, pp.
  295--303, Dec. 1983.

\bibitem{Mei94}
W.~Mei, ``Families of nonrepeating {FH} sequences with given minimum gap,''
  \emph{Journal of China Institute of Communications}, vol.~15, no.~6, pp.
  63--68, Jun. 1994.

\bibitem{HY09}
H.~Wang, Y.~Zhao, F.~Shen, and W.~Sun, ``The design of wide interval {FH}
  sequences based on {RS} code,'' in \emph{2009 International Conference on
  Mechatronics and Automation}.\hskip 1em plus 0.5em minus 0.4em\relax
  Institute of Electrical and Electronics Engineers (IEEE), Sep. 2009, pp.
  2345--2350.

\bibitem{PCY19}
P.~Li, C.~Fan, Y.~Yang, and Y.~Wang, ``New bounds on wide-gap frequency-hopping
  sequences,'' \emph{IEEE Commun. Lett.}, vol.~23, no.~6, pp. 1050--1053, Jun.
  2019.

\bibitem{QH22}
Q.~Shu, H.~Liu, X.~Liu, Y.~Yang, and W.~Chen, ``Optimal wide-gap-zone frequency
  hopping sequences,'' \emph{Adv. Math. Commun.}, vol.~18, no.~5, pp.
  1379--1389, Nov. 2024.

\bibitem{LC22}
L.~Zhou, C.~Zhang, Q.~Zeng, X.~Liu, and H.~Wu, ``Optimal low-hit-zone
  frequency-hopping sequence sets with wide-gap for {FHMA} systems under
  follower jamming,'' \emph{IEEE Commun. Lett.}, vol.~26, no.~5, pp. 969--973,
  May 2022.

\bibitem{PC22}
P.~Li, C.~Fan, S.~Mesnager, Y.~Yang, and Z.~Zhou, ``Constructions of optimal
  uniform wide-gap frequency-hopping sequences,'' \emph{IEEE Trans. Inf.
  Theory}, vol.~68, no.~1, pp. 692--700, Jan. 2022.

\bibitem{XC25}
X.~Zheng, C.~Fan, Z.~Zhou, S.~Mesnager, and Y.~Yang, ``Wide-gap frequency
  hopping sequences with no-hit-zone: Bounds and their optimal constructions,''
  \emph{IEEE Trans. Inf. Theory}, vol.~71, no.~5, pp. 3989--3998, May 2025.

\bibitem{XC20}
X.~Niu, C.~Xing, and C.~Yuan, ``Asymptotic {Gilbert-Varshamov} bound on
  frequency hopping sequences,'' \emph{IEEE Trans. Inf. Theory}, vol.~66,
  no.~2, pp. 1213--1218, Feb. 2020.

\bibitem{XC10}
X.~Tang and C.~Ding, ``New classes of balanced quaternary and almost balanced
  binary sequences with optimal autocorrelation value,'' \emph{IEEE Trans. Inf.
  Theory}, vol.~56, no.~12, pp. 6398--6405, Dec. 2010.

\bibitem{WF23}
W.~Ren and F.~Wang, ``A new class of optimal wide-gap one-coincidence
  frequency-hopping sequence sets,'' \emph{Adv. Math. Commun.}, vol.~17, no.~2,
  pp. 342--352, Apr. 2023.

\bibitem{Cao2006Nov}
Z.~Cao, G.~Ge, and Y.~Miao, ``{Combinatorial characterizations of
  one-coincidence frequency-hopping sequences},'' \emph{Des., Codes Cryptogr.},
  vol.~41, no.~2, pp. 177--184, Nov. 2006.
\end{thebibliography}
\end{document}